\documentclass[sigconf,edbt]{acmart-edbt2020}
\usepackage{multirow}
\usepackage{xcolor}
\usepackage{url} 
\usepackage{hyperref}
\usepackage{comment}
\usepackage{amsmath}

\def\BibTeX{{\rm B\kern-.05em{\sc i\kern-.025em b}\kern-.08em
    T\kern-.1667em\lower.7ex\hbox{E}\kern-.125emX}}

\usepackage{booktabs} % For formal tables
\usepackage{epsfig}
\usepackage{subcaption}
\begin{document}

\fancyhead{}

\title{SigGAN : Adversarial Model for Learning Signed Relationships in Networks}


\author{Roshni Chakraborty}
\affiliation{%
  \institution{Indian Institute of Technology Patna}
  \city{Patna}
  \country{India}}
\email{roshni.pcs15@iitp.ac.in}

\author{Ritwika Das}
\affiliation{%
  \institution{National Institute of Technology Durgapur}
  \city{Durgapur}
  \country{India}
}
\email{ritwikadas1@gmail.com}

\author{Joydeep Chandra}
\affiliation{%
 \institution{Indian Institute of Technology Patna}
 % \city{}
  \country{India}}
\email{joydeep@iitp.ac.in}

\renewcommand{\shortauthors}{Chakraborty et al.}

\begin{abstract}
  Signed link prediction in graphs is an important problem that has applications in diverse domains. It is a binary classification problem that predicts whether an edge between a pair of nodes is positive or negative. Existing approaches for link prediction in unsigned networks cannot be directly applied for signed link prediction due to their inherent differences. Further, additional structural constraints, like, the structural balance property of the signed networks must be considered for signed link prediction. Recent signed link prediction approaches generate node representations using either generative models or discriminative models. Inspired by the recent success of Generative Adversarial Network (GAN) based models which comprises of a discriminator and generator in several applications, we propose a Generative Adversarial Network (GAN) based model for signed networks, SigGAN. It considers the requirements of signed networks, such as, integration of information from negative edges, high imbalance in number of positive and negative edges and structural balance theory. Comparing the performance with state of the art techniques on several real-world datasets validates the effectiveness of SigGAN.\\
\end{abstract}

\maketitle

%%However, the major problem with these deep neural approaches is that for high dimensional data they often suffer from poor compatibility with the internal neural network structure. Further, the deterministic bias resulting from the finite samples of sub-networks used for training along with the missing data is also difficult to handle using these models. 

\section{Introduction}

\par The huge popularity of online social networks attracts users to share and exchange their views and opinions. Existing research works represent these relationships between the users through signed networks~\cite{yang2007community,yang2017stochastic} where the edges depict both polarity and strength~\cite{yuan2017negative,derr2018relevance}. Understanding these relationships has many applications~\cite{salehi2018detecting,das2018positive,tang2016node}, like, personalized product recommendation~\cite{massa2007trust} and identification of malicious users~\cite{kumar2014accurately}. Therefore, the prediction of the signed links in these networks remains an important problem~\cite{wu2016troll}. However, prediction of links in signed networks has different challenges than the 
conventional unsigned networks~\cite{beigi2016signed}. These challenges mainly stem from the differences due to the inclusion of sign in links. For example, several theories of social science, like \textit{homophily} and \textit{social influence}~\cite{shalizi2011homophily,doreian1989network} does not apply directly to signed networks. While two users who are connected are likely to be similar in unsigned networks, negative link between two users means they are likely to be dissimilar in signed networks whereas positive link between two users signify that they are likely to be similar. So, in signed networks, \textit{homophily} and \textit{social influence} does not directly apply just on the basis of the presence of link between a pair of nodes but also depend on the sign of the link. Similarly, the presence of sign affects the nature of transitive relationship between three users~\cite{tang2014ht}, like, \textit{enemy of my friend is my enemy}, \textit{friend of my friend is my friend} and \textit{enemy of my enemy is my friend} are common in signed networks whereas in unsigned networks, two nodes connected to a common node has a higher likelihood to be connected among themselves than two nodes who do not share a common node. Therefore, in signed networks, both the presence of sign and the link are important in order to understand the type of transitive relationship. Furthermore, Leskovec et al.~\cite{leskovec2010predicting} highlights that the presence of negative links adds more value as compared to only positive links. However, it is highly challenging to integrate the information from the negative links as the fraction of negative links is very low compared to the positive links.

\par To mitigate these issues, existing research works propose traditional feature based learning~\cite{javari2017statistical,tang2015negative,chen2019link,chakraborty2020we} and recently, deep learning approaches. Existing deep learning approaches rely on different forms of node embedding approaches which does not require the feature engineering efforts of traditional feature based learning approaches. These approaches either follow \textit{generative} or \textit{discriminative} based models to derive the node embedding. The \textit{generative} model-based representation learning frameworks rely on the conditional distribution of a node, $v_b$, being a neighbor of node $v_a$, given the connectivity preferences of $v_a$. These frameworks learn the likelihood of edges based on this conditional connectivity distribution, and thereby generate the vertex embedding that maximizes the likelihood of the given network connectivity. Several existing works, like SNE~\cite{yuan2017sne}, SIDE~\cite{kim2018side}, SigNET~\cite{islam2018signet}, sign2Vec~\cite{bhowmick2019network} follow \textit{generative} based approach to generate signed network embeddings. Subsequently, several existing research works, like SiNE~\cite{wang2017signed}, SHINE~\cite{wang2018shine}, DNE-SBP~\cite{shen2018deep}, have proposed \textit{discriminative} models which predict the presence of an edge based on the characteristics of the node 
pair, $v_a$ and $v_b$. 
However, there are several issues like missing and incomplete data that are quite prevalent in large scale network data that affect the performance of the  generative or discriminative based models~\cite{goodfellow2017nips}. GAN based models by combining both generative and discriminative models have shown to be more robust to these issues~\cite{goodfellow2017nips}. Recently, they have been highly successful in several applications, like, information retrieval~\cite{lu2019psgan}, image generation~\cite{denton2015deep} and collaborative filtering~\cite{chae2019rating,chae2018cfgan}. Recently, Wang et al. ~\cite{wang2018graphgan} have proposed a GAN based model to represent nodes in unsigned networks in which the generator effectively learns the underlying connectivity distribution which is highly effective for link prediction in unsigned networks. 

%However, this cannot be directly applied to signed networks, where the presence of links between a pair of nodes represents both similarity and dissimilarity of the node pair. Further, one of the major characteristics of the signed networks is the structural balance theory [5-7] among the connected node triplets.

%However, issues like missing and incomplete data that are quite prevalent in large scale network data affect the performance of these models whereas generative adversarial networks (GAN) are known to be more robust~\cite{goodfellow2017nips}. GAN based models combine both \textit{generative} and \textit{discriminative} models and have been highly successful in several applications, . Recently, Wang et al.~\cite{wang2018graphgan} have extended the usage of GAN based models to networks by a GraphGAN model which effectively learns the underlying connectivity distribution of nodes for node embedding. 

%%However, the major problem with these deep neural approaches is that for high dimensional data they suffer from poor compatibility with the internal neural network structure. Further, they suffer from deterministic bias resulting from the finite samples of sub-networks used for training. 

\par  However, GraphGAN~\cite{wang2018graphgan} cannot be directly applied to signed networks, due to the inherent differences between signed and unsigned networks as discussed previously. Further, one of the major characteristics of the signed networks is the structural balance theory~\cite{heider1958psychology,cartwright1956structural,hummon2003some} among the connected node triplets. In a structurally balanced triplet, the sign of the link between a pair of nodes $(v_a, v_b)$ connected to a common neighbor $v_c$ would be positive if the sign of the links $(v_a, v_c)$ and $(v_b, v_c)$ are either both positive or negative. Similarly, the sign of link $(v_a, v_b)$ would be negative if the sign of one of the links to $v_c$ is positive, and the other is negative. Therefore, we propose SigGAN, a GAN based model which is adapted to signed networks and integrates structural balance property.
The novel contribution of this paper is a GAN based model for generating node representations for signed networks that considers both the structural proximity of the nodes and the structural balance of the network. The proposed framework effectively unifies two models --- a binary-class \textit{generative model} that generates likely positive (or, negative), connected neighbors of a node, $v_a$ and a \textit{discriminative model} that identifies whether the generated node pair with the corresponding sign is one from the generated samples or the actual dataset --- to learn individual parameters for each model to minimize a combined loss function.
Validation on real world datasets, like \textit{Slashdot}, \textit{Epinions}, \textit{Reddit}, \textit{Bitcoin} and \textit{Wiki-RFA} indicate that the SigGAN can ensure $3-11\%$ higher micro F1-score than the existing works in both \textit{link prediction} and \textit{handling of sparsity}. The organization of the paper is as follows. We discuss the existing research works in Section~\ref{s:rel}. A formal definition of the problem is provided in Section~\ref{s:prob} followed by the details of the SigGAN in Section~\ref{s:prop}. We discuss the experimental settings in Section~\ref{s:expt} and the details of our observations and results are given in Section~\ref{s:res}. We finally draw our conclusions in Section~\ref{s:con}.

\section{Related Works}~\label{s:rel}
%Several research works aim to predict the sign of a link in signed networks. One of the possible categorizations of these works can be based on the learning techniques used. 

Existing research works which predict the sign of a link can be categorized into feature based approaches and representation learning based approaches. We discuss the works that belong to either of these categories next.

\subsection{Feature based Machine Learning Approaches}

Existing feature based machine learning approaches use either \textit{local}, \textit{global} or a comibination of both to predict the sign of the link. For example, the \textit{local} attributes include information about common neighbours, structural balance~\cite{leskovec2010predicting,leskovec2010signed}, node features~\cite{tang2015negative}, etc. Existing approaches have explored the role of 
user interactions~\cite{agrawal2013link}, similarity in user characteristics~\cite{cheng2017unsupervised}, neighborhood information~\cite{hu2017sparse}, etc., to predict the sign of the missing link. While the utilization of the \textit{local} attributes incur low computational cost, it fails in sparse networks. While \textit{global} attributes, like, social hierarchy~\cite{lu2016exploring}, individual trust~\cite{wu2016troll, guha2004propagation}, structural balance theory~\cite{cartwright1956structural} and community based ranking~\cite{shahriari2014ranking} can handle the sparsity, they are computationally expensive~\cite{khodadadi2017sign}. To mitigate this computation cost, several research works explore \textit{meso} attributes, like, community structures~\cite{khodadadi2017sign,chiang2014prediction,chiang2011exploiting}, edge dual properties~\cite{yuan2017edge}, Katz similarity-based walk~\cite{singh2017measuring} and, clustering-based approaches~\cite{javari2014cluster}. However, utilization of only one type of attribute affects the prediction accuracy of the sign of the link.

%extraction of features while drilling through the entire network still incurs high computational cost. Further, utilization of the only \textit{global} attributes fails to understand the neighborhood relationships which negatively affects the accuracy in prediction of the sign of the link.

\par Therefore, recent works explore a combination of \textit{local} and \textit{global} attributes~\cite{chen2017link,hu2017sparse,javari2017statistical}. Though a combination of both attributes can ensure higher accuracy, the extraction of \textit{global} features for all nodes will lead to high computation cost. Further, these existing works ignore the high variance in the network characteristics of the edges that may provide vital cues in the prediction of the sign of an edge~\cite{chakraborty2020we}. This knowledge can aid in the development of an adaptive system that applies \textit{local}, \textit{global}, or a combination of both based on the characteristics of the edge~\cite{chakraborty2020we}. Further, these approaches are highly dependent on the selection of the features and therefore, deep learning approaches which can inherently integrate both \textit{local} and \textit{global} node attributes are more successful. 
%We provide a brief overview of the existing node embedding based approaches hereby.

\subsection{Node Embedding Based Approaches}
%Although there are a plethora of research works~\cite{zhang2018network} for representation learning (node embedding) for unsigned networks, the SigGANes are not directly applicable to signed networks. This is mainly due to the inherent differences in both of these networks, like different \textit{transitivity} and \textt{homophily} relationships and the presence of structural constraints based on structural balance and status theory.  We discuss the relevant literature for node embedding in signed networks.Network embedding is one of the fundamental problems in network analysis and has received much interest from the data mining community recently [1, 2, 4, 7, 9, 21, 23, 28–31, 33, 34, 37]. Network embedding maps nodes into low-dimensional vector space that summarizes various aspects of network topol- ogy and link structure.

Recently, node embedding has been proven to be very useful in understanding the structure of the network, link prediction, node classification, etc. Several existing research works~\cite{perozzi2014deepwalk,grover2016node2vec,tang2015line,tian2014learning,wang2016structural} have explored different aspects of network topology of a node to map each node into a low-dimensional vector space. However, these approaches can not be directly applied to signed networks due to the inherent differences between signed and unsigned networks. Existing node embedding approaches for signed networks are either \textit{matrix factorization} based or, \textit{deep learning frameworks}. The \textit{deep learning frameworks} can further be segregated into \textit{generative} and \textit{discriminative} based representation learning, where \textit{generative} models focus on learning the underlying connectivity distribution in the graph and \textit{discriminative} models predict the probability of an edge based on the attributes of the pair of nodes. Existing \textit{generative} models rely on \textit{representation learning by random walk based approaches} and the \textit{discriminative} models rely on \textit{representation learning by adjacency relationship-based embedding} which we discuss next.  

\par Existing \textit{representation learning by random walk based approaches} consider the nodes present in randomly selected paths~\cite{yuan2017sne} or, the co-occurrence of nodes from randomly selected paths~\cite{kim2018side}. Further, to handle the high computational cost to incorporate all the possible edges, Kim et al.~\cite{kim2018side} propose a variant of \textit{negative sampling} in signed networks. However, none of these works consider the higher-order proximity or the signed relationships. Bhowmick et al.~\cite{bhowmick2019network} address this issue by a trust-based random walk that considers both  higher-order neighborhood and signed (trust) relationships. However, as none of these works explicitly consider structural balance into the generated embeddings. Therefore, Islam et al.~\cite{islam2018signet} introduce the integration of \textit{structural balance} into random walks. Subsequently, \textit{discriminator} based frameworks generate low dimension representation of nodes, like, Shen et al.~\cite{shen2018deep} explicitly considers structural balance information while generating the embedding vector of a node. Additionally, Lu et al.~\cite{lu2019ssne} consider both \textit{status} and \textit{structural balance} and Chen et al.~\cite{chen2018bridge} consider \textit{bridge edges}, \textit{status} and \textit{structural balance} for generating the node embedding. However, the major problem with all these methods is that they suffer from poor compatibility for high dimensional data. Further, they are prone to poor training in the face of missing edges or smaller sample sizes~\cite{goodfellow2017nips}, which are major issues in large scale networks generated from real data.

\par On the other hand generative adversarial networks are more robust to these issues~\cite{goodfellow2014generative} and hence, have been used extensively in several applications~\cite{zhang2018generative} including predicting links in unsigned networks~\cite{wang2018graphgan}. Unlike unsigned networks, where the node representations mainly deal with presence or absence of links, for signed networks, the problem is identifying the sign of the link. Further, for signed networks, the GAN model should be able to handle specific properties of signed networks, like the high imbalance in positive and negative links and \textit{structural balance}. In this paper, we propose a GAN model for generating node representations in signed networks that can handle these issues. We next provide a formal definition of the problem and subsequently, discuss the SigGAN.

%While the negative sign $s^{\prime}$ for an edge between $v_i$ and $v_j$ indicates a possible distrust or enemy relationship between $v_i$ and $v_j$, the positive sign, $s$, indicates a trust or friendship relation between the nodes.  However, to the best of our knowledge, GAN frameworks have never been directly applied to signed networks. 

\section{Problem Definition}~\label{s:prob}

\par Let, $G=(V, E, S)$ denote a signed network, where $V$ denotes the set of vertices
(nodes) and $E$ represents the set of edges which has either positive ($s$) or negative ($s^{\prime}$) sign. $S=\{s,s^{\prime}\}$ denotes the set of signs considered for the edges. We represent the signs as $s$ and $s^{\prime}$ rather than $+$ or $-$. We assume that we have information about certain edges along with their signs for a network. Thus, for a given node $v_j$, the objective of the problem is to predict the possible sign ($s$ or $s^{\prime}$) of the missing links from $v_j$. Therefore, for a given node $v_j$ we determine the probability $p(\phi_i|v_j)$, for all nodes $v_i (i\neq j)$, where $\phi_i$ represents a random variable that takes values in $S$. For example, $p(\phi_i=s|v_j)$  denotes the probability that $v_i$ is positively connected to $v_j$. We can further extend this problem to directed networks, i.e., an edge from $v_i$ to $v_j$ and $p(\phi_i|v_j)$ may not be equal to $p(\phi_j|v_i)$. However, this would also involve implementing suitable discriminator models that also consider the directions of the edges. Since the objective of this paper is to show the applicability of GAN in generating node representations in signed networks, we avoid dealing with the issue of the direction of the edges in the network. 
The formation of the signed network strictly depends on the application. For example, for a social news site, like \textit{Slashdot} \footnote{https://snap.stanford.edu/data/soc-Slashdot0902.html}, the signed network comprises of users as nodes and the edges are ($-1$, or $+1$) depending on the interaction of the users in comments. The prediction task can be considered as determining the probably signed rating of one user towards another. We discuss SigGAN next. 
%a signed network  formed with the users, commenting on news, as nodes and a directed edge is formed from one user to another if the former has rated ($-1$, $0$ or $+1$) on the comment of the other. The ratings can be made based on typical perceptions about the comment like normal, off-topic, insightful, redundant, interesting, or troll to name a few. 

\section{Proposed Approach}~\label{s:prop}
SigGAN relies on a GAN based framework that combine a generative and discriminative model using a minimax game~\cite{wang2018graphgan} to generate node embedding. The generative model in GAN either outputs a labeled data from the true dataset or generates a fake labeled data based on the node embedding, whereas the discriminator model attempts to classify the generated data as true or fake. In SigGAN, we explicitly consider the characteristics of signed networks, like \textit{structural balance theory} property~\cite{islam2018signet} , \textit{signed homophily} and the usual requirements that have been considered in GraphGAN~\cite{wang2018graphgan}, such as, \textit{structural awareness} and \textit{low computational complexity}.

%To the best of our knowledge, this is the first effort in using GANs for signed link prediction that considers all these requirements. 

%Although node embedding generated through GAN based model have been applied for unsigned link prediction (like GraphGAN~\cite{wang2018graphgan}), however as stated earlier, the typical characteristics of signed networks pose new challenges in applying these models for signed link prediction. Unlike unsigned networks where connected nodes will have similar embedding in the representation space, in signed networks, the separation of the embedding should depend on the sign of the links. Two nodes connected by a negative link has a higher separation than the nodes connected by a positive link. Further, as previously discussed, \textit{structural balance theory} property~\cite{islam2018signet} which must also be considered in generating the representations. 

\par SigGAN learns the underlying signed adjacency distribution  of each node, $v_j$ across all nodes in $V$, i.e., $p_{true}(\phi_i|v_j)$, where  $\phi_i\in S$ represents the sign of the link from $v_j$ to $v_i$. SigGAN uses a generator and discriminator model that operates at tandem to improve the prediction accuracy of the system. The task of the generator function $J(\phi_i|v_j;\theta_J)$ is to approximate the true positive and negative connections of $v_j$, that is used to select it's likely neighbors with a given sign from the underlying distribution $p_{true}(\phi_i|v_j)$. The proposed generator model has the following key properties:
\begin{enumerate}
    \item It generates a probability of the sign of the edges from $V_i$ to other nodes.
    \item It is structurally aware, where the probability of the existence of an edge to a far away node tends to zero.
    \item It maintains the structural balance theory.
    \item The generator model is computationally efficient.
\end{enumerate}
The discriminator function, $D(v_i,v_j,\phi_i;\theta_D)$ uses a MLP based model that outputs a single scalar value that represents how likely an edge $(v_i,v_j)$ is of a given sign $\phi_i$. We train the discriminator by an iterative selection of negative and positive edges. However, Shen et al. \cite{shen2018deep} shows that the number of negative edges is very less compared to the number of positive edges and the cost of forming a negative edge is higher than the formation of a positive edge in signed networks. To handle this, we select negative edges with an equal probability as the positive edges for the discriminator to learn so that SigGAN learns to predict the negative edge with high precision as well.
Therefore, given the graph $G$, generator $J$ tries to generate the positive (or negative) neighbors of $v_j$ which are similar to its actual positive (or negative) neighbors. We term these neighbors of $v_j$ generated by the generator as \textit{fake} neighbors, whereas those neighbors from the actual graph are \textit{true} neighbors. The discriminator $D$ determines if the generated neighbors of $v_j$ are fake or true through a two player minimax with value function $\mathcal{V}(J, D)$, given as
\begin{eqnarray}
\min_{\Theta_J} \max_{\Theta_D} \mathcal{V}(J,D)=\sum_{j=1}^{V}(\mathbb{E}_{v \sim p_{true}(\cdot|v_{j})}[\log D(v, v_{j}, \phi;\theta_{D})]\nonumber\\ + \mathbb{E}_{v \sim J(\cdot|v_{c}; \theta_{J})}[\log (1- D(v, v_{j},\phi;\theta_{D})]) \label{eq:valFunc}
\end{eqnarray}

$\theta_D$ represents the vector representations of all nodes $v_i$, $i\neq j$, (described later) which can be potential neighbors of $v_j$. Therefore, the optimal parameters of the generator and discriminator are learned by alternately maximizing and minimizing the value function $\mathcal{V}(J, D)$. At each iteration, the discriminator $D$ is trained using a batch of true samples from $p_{true}(\cdot|v_j)$ and fake samples generated using $J(\cdot|v_j)$. The generator $J$ is updated using a policy gradient. The continuous competition between $D$ and $J$ ensures generating a suitable representation of the node, $v_j$ that satisfies the required properties. We discuss the functions of $D$ and $J$ next.
\subsection{Discriminator}
Given a list of node pairs ($\cdot, v_j)$, the objective for
the discriminator, $D$ is to maximize the log-probability of assigning the correct label depending on whether the edge has been sampled from the actual network (true) or has been generated by the generator (fake). The discriminator mainly relies on node representations that can be compared pairwise to discriminate between true and fake edges. 

%\subsubsection{The discriminator model} 
\par There exists a plethora of models that can be used for the discriminator; in this paper we do not contribute in investigating the suitability of the models and leave it as a future goal. Rather, we use a dense network to determine node representations that can be compared to derive the likely chances of them being connected with a given sign $\phi$. The model uses a single hidden layer, in which the number of units equals the dimension of the node representations (represented by $k$ in our case). The input to the model is the one-hot vector of a node in the network. Each unit in the output layer corresponds to a node in the network and is activated by a sigmoid function. The output of the sigmoid function of the $i^{th}$ output unit provides a measure of the possibility of node $v_i$ being connected to $v_j$ with sign $\phi_i$.
Suppose, $h_j\in\mathbb{R}^{n}$ is the one hot encoding vector of node $v_j$ that is provided as input to the input layer of the discriminator. Let, $W\in\mathbb{R}^{n\times k}$ represent a weight matrix whose $i^{th}$ column elements represent the corresponding weights of the links from each of the inputs to the $i^{th}$ unit of the hidden layer. Thus, $d_j=W^{T}h_j$ represents the hidden unit vector of length $k$. Suppose, $W^{\prime}\in \mathbb{R}^{k\times n}$ represents a matrix of the weights of hidden units to the output layer, where the $[ij]^{th}$ element of $W^{\prime}$ represents the weight of the link connecting the $i^{th}$ hidden layer unit to the $j^{th}$ output unit. Thus, the output $o_i$ of each of the output units can be represented as $o_i=\sigma(W^{\prime}_{\cdot i}\cdot d_j)$, where $W^{\prime}_{\cdot i}$ represents the $i^{th}$ column of $W^{\prime}$, and can be considered as a $k$ dimension embedding of the node $v_i$, represented as $d_i$. Further, $\sigma$ is the sigmoid function given as $\sigma(x)=\frac{1}{1+\exp(-x)}$. 
Thus, if $d_i, d_j\in \mathbb{R}^k$ represents the $k$ dimensional embedding vector for nodes $v_i$ and $v_j$, respectively, and $\phi$ represents the corresponding link sign of these nodes, then the discriminator function would be represented as 
\begin{eqnarray}
D(v_i, v_j, \phi)=\sigma(\phi d_i^Td_j)
=\frac{1}{1+\exp(-\phi d_i^Td_j)}\label{eq:discrim}
\end{eqnarray}
As evident from equation~\ref{eq:discrim}, when $\phi$ is positive, $D$ returns a near to $1$ value if $d_i$ and $d_j$ are similar, indicating that nodes with similar representations are most likely to be positively connected. On the other hand for a negative value of $\phi$, the discriminator would return a higher score if $d_i$ and $d_j$ are highly dissimilar. The discriminator is trained based on the loss function stated in equation~\ref{eq:valFunc}. For each $d_j$, a stochastic gradient ascent is used to update $\theta_D$, the vector representations of $v_i$. We describe the generator model that we use in our SigGAN next.

\subsection{Generator}
Contrary to the discriminator, the goal of the generator $J$
is to generate the links of each node $v_j$ with their corresponding sign that mimics the true adjacency distribution $p_{true}(\phi_i|v_j)$. The generator model is defined as a function $J(\phi_i|v_j)$ that generates node representations to minimize the loss function
\begin{eqnarray}
\mathbb{E}_{v\sim J(\cdot|v_j;\theta_{J})}[\log(1-D(v_i,v_j,\phi;\theta_D)],
\end{eqnarray}
which is the log probability of the discriminator correctly identifying fake link samples generated by the generator. Since the generator samples through the discrete space $v_i$ to increase the probability score, to minimize the loss function we derive its gradient with respect to $\theta_J$ using policy gradient. Thus, if $N(v_j)$ denotes the neighbors of $v_j$, then 
\begin{eqnarray}
& &\nabla_{\theta_J}\sum_{j=1}^{|V|}\left[\mathbb{E}_{v\sim J(\cdot|v_j;\theta_J)}[\log(1-D(v, v_j,\phi;\theta_D)]\right]\nonumber\\
&=&\nabla_{\theta_J}\sum_{j=1}^{|V|}\sum_{i=1}^{N(v_j)}J(\phi_i|v_j)[\log(1-D(v_i, v_j,\phi_i;\theta_D)]\nonumber\\
&=& \sum_{j=1}^{|V|}\sum_{i=1}^{N(v_j)}\nabla_{\theta_J} J(\phi_i|v_j)[\log(1-D(v_i, v_j,\phi_i;\theta_D)]
\nonumber\\
&=&\sum_{j=1}^{|V|}\sum_{i=1}^{N(v_j)} J(\phi_i|v_j)\times\nonumber\\
& &\left[\nabla_{\theta_J} \log(J(\phi_i|v_j))\right] \left[\log(1-D(v_i, v_j,\phi_i;\theta_D))\right]\label{eq:genLoss}
\end{eqnarray}

Therefore, we can intuitively understand from Equation~\ref{eq:genLoss} that a higher value of $\log(1-D(v_i,v_j,\phi_i;\theta_D))$ for a given $v_i$ will lower the probability of generating $v_i$ with respect to $v_j$. We use a softmax function for the generator $J(\phi_i|v_j)$ which is given as
\begin{eqnarray}
J(\phi_i=s|v_j))&=&\frac{e^{s\left(g_i^Tg_j\right)}}{\sum_{k\neq j} e^{s\left(g_k^Tg_j\right)}}\label{eq:genSoft}
\end{eqnarray}
As evident from Equation~\ref{eq:genSoft}, for a given node $v_j$, $J$ generates a  neighbor $v_i$ with sign $s$ based on the representations, $g_i$ and $g_j$ of the nodes $v_i$ and $v_j$, respectively. A similar representation between the two nodes increases the probability of them being positively connected, whereas dissimilar representations increase the chances of a negative link. These representations are suitably derived based on the loss function stated in Equation~\ref{eq:genLoss} using gradient descent. However, a major problem with the softmax function is the high computation involved in updating the gradients. As evident from Equation~\ref{eq:genSoft}, for an arbitrary node $v_i$, the gradients must be calculated for the entire set of vertices and hence, $J(\phi_i|v_j)$ would require updating the representations of all $|V|$ nodes. Further, the softmax function ignores the rich structural information of the graphs, where the proximity of the nodes in the network can play an important role in the link formation~\cite{wang2018graphgan}. Existing approximation techniques that determine the softmax scores, like negative sampling and subsampling~\cite{mikolov2013distributed} do not consider the structural information. Hence, there is a need to address these issues for the practical applicability of this method. We propose a modification of the softmax approach that can mitigate these issues. 
\subsection{Modified Softmax for Signed Graphs} \label{s:modSmax}
In this Section, we discuss the proposed variant of the softmax function that can be used for the signed graphs. Although in GraphGAN~\cite{wang2018graphgan} a modified softmax function was proposed for unsigned networks that maintain specific properties, like, normalization, graph structure awareness, and computational efficiency, the proposed function is not applicable for signed networks. Further, an additional requirement of the signed networks is the structural balance among the nodes, which the softmax function must consider. These requirements demand the formulation of a suitable softmax function applicable for signed networks. We, hereby, summarize the four basic properties that we consider:
\begin{enumerate}
    \item \textit{Normalization} : The function must be normalized such that it produces a valid probability distribution, i.e., $\sum_{i\neq j}\sum_{t\in S}J(\phi_i=t|v_j)=1$.
    \item \textit{Graph Structure aware}: The function must take into account the structure of the graph to calculate the connectivity distribution between a pair of nodes.
    \item \textit{Computationally Efficient} : The updation of the values of $\theta_J$ for deriving the gradient descent must be computationally efficient.
    \item \textit{Structurally Balanced} : The sign of the links generated must maintain Structural balance theory.
\end{enumerate}
We next discuss the proposed softmax approach. Given a node, we introduce \textit{sign-specific relevance} probability of its neighbor. For a node $v_j$, the positive relevance probability of its neighbor $v_{i}$ is given as
\begin{eqnarray}
p(\phi_i=s|v_j)=\frac{e^{s(g_j^Tg_i)}}{\sum_{k\in N(v_j)}\sum_{t\in S} e^{t(g_j^Tg_k)}},\label{eq:signRelev}
\end{eqnarray}
whereas the negative relevance probability is obtained by replacing $s$ by $s^{\prime}$.
The sign-specific relevance probability indicates how likely a neighbor $v_i$ of node $v_j$ would be connected by a link with a given sign,  $s$ or $s^{\prime}$. Similar to the GraphGAN approach, for each node $v_j$, we initially make a Breadth First Search (BFS) traversal with $v_j$ as the root node. Let $v_{r_0}\rightarrow v_{r_1} \rightarrow \cdots\rightarrow v_{r_m}$ be a path in the BFS tree with root at $v_{r_0}$.   Let $\pi^{n-1}_{n,s}=p(\phi_{r_n}=s|v_{r_{n-1}})$ be the positive relevance probability of  node $v_{r_n}$ with respect to its uplink parent $v_{r_{n-1}}$. 
Therefore, $\pi^{0}_{m,s}$ is the relevance probability of the node at hop $m$ (written as sub-script of $\pi$) from the node at hop $0$ (written at super-script of $\pi$) connected by an edge with positive sign, $s$. From \textit{structural balance theory}, we can derive the sign of the edge connecting node at level $0$ and $m$ through the level, $m-1$. For example, on the basis of \textit{structural balance theory} for triads, the sign of the edge between node at level $0$ and $m$ would be positive, $s$, if the sign of the edge node at level $0$ and $m-1$ is $s$ and the sign of the edge node at level $m-1$ and $m$ is $s$. Additionally, the sign of the edge between node at level $0$ and $m$ would also be positive if the sign of the edge node at level $0$ and $m-1$ is $s^{\prime}$ and the sign of the edge node at level $m-1$ and $m$ is $s^{\prime}$. Similarly, we can find the the sign of the edge between node at level $0$ and $m$ would be negative if the sign of the edge between node at level $0$ and $m-1$ and between $m-1$ and $m$ is different. Therefore, we can extend the positive and negative relevance probability, respectively, for nodes that are more than one hop away by recursive formulations given as 
\begin{eqnarray}
\pi^{0}_{m,s}&=&
(\pi^0_{m-1,s}\pi^{m-1}_{m,s}+\pi^{0}_{m-1,s^{\prime}}\pi^{m-1}_{m,s^{\prime}}) \nonumber\\ \pi^{0}_{m,s^{\prime}}&=&(\pi^0_{m-1,s}\pi^{m-1}_{m,s'}+\pi^{0}_{m-1,s'}\pi^{m-1}_{m,s})\label{eq:recursivePi}
\end{eqnarray}
We are interested in determining the modified softmax function $J(\phi_{r_n}|v_{r_0})$ for any arbitrary node $v_{r_n}$ in the tree path with respect to the root node $v_{r_0}$. Assuming that the graph is strongly connected, a BFS tree with respect to a node as root will traverse all the nodes through a single unique path. This assumption holds true for any undirected graph; several real-world directed networks are also strongly connected with a large core and hence can be applied to these. Thus, we define $J(\phi_{r_n}|v_{r_0})$ as
\begin{eqnarray}
J(\phi_{r_n}=s|v_{r_0})&=&\pi^0_{n,s}\pi^n_{n-1,s} + \pi^0_{n, s^{\prime}}\pi^n_{n-1,s'}\nonumber\\
J(\phi_{r_n}=s^{\prime}|v_{r_0})&=&\pi^0_{n,s}\pi^n_{n-1,s^{\prime}} + \pi^0_{n, s^{\prime}}\pi^n_{n-1,s}\label{eq:genRoot}
\end{eqnarray}
We prove using appropriate theorems that the generator function $J(\phi_{r_n}=s|v_{r_0})$ satisfies the first three requirements that are mentioned above, whereas, we intuitively show that the modified softmax can also learn the structural balance property if the same exists in the network.
\begin{theorem}
For a given node $v_j$, using the modified softmax we get $\sum_{i\neq j}\sum_{t\in S}J(\phi_i=t|v_j)=1$
\end{theorem}
%\\$\sum_{v_i\in ST_m}J(\phi_i=t|v_{r_0})=\sum_{t\inS}J(\phi_{r_m}=t|v_{r_0})$. In other words,  
\begin{proof}
We start by showing that for any sub-tree $ST_m$ (rooted at $v_{r_m}$) of the BFS tree with $v_{r_0}$ as the root, the total sum of the softmax scores of all the nodes in the sub-tree, with respect to $v_{r_0}$,  is the softmax score of $v_{r_m}$ calculated with respect to $v_{r_0}$. We use this concept for the sub-trees rooted at the child nodes of the BFS tree of $v_{r_0}$. From the expression of sign-specific relevance probability stated in Equation~\ref{eq:signRelev}, we can directly conclude that for a given node $v_j$, if $N(v_{j})$ be its neighbors then $\sum_{v_i\in N(v_j)}\sum_{t\in S}p(\phi_i=t|v_j)=1$.
\par Initially, we consider the case where the subtree is rooted at the node $v_{r_m}$, which has only leaf nodes as its children. Let the children be denoted as $v_{cm_1}, v_{cm_2},\ldots,v_{cm_l}.$
So, $$\sum_{i\in ST_{r_m}}\sum_{t\in S}J(\phi_i=t|v_{r_0})$$ can be written as
$$\sum_{t\in S}J(\phi_{r_m}=t|v_{r_0}) + \sum_{i\in C(v_{r_m})}\sum_{t\in S}J(\phi_{cm_i}=t|v_{r_0}),$$
where $C(v_{r_m})$ denotes the child nodes of $v_{r_m}$.
From Equation~\ref{eq:recursivePi}, we have
\begin{eqnarray}
& &\sum_{t\in S}J(\phi_{r_m}=t|v_{r_0})\nonumber\\
&=& (\pi^0_{m-1,s}\pi^{m-1}_{m,s}+\pi^0_{m-1,s^{\prime}}\pi^{m-1}_{m,s^{\prime}})\pi^{m}_{m-1,s}\nonumber\\
&+& (\pi^0_{m-1,s}\pi^{m-1}_{m,s^{\prime}}+\pi^0_{m-1,s^{\prime}}\pi^{m-1}_{m,s})\pi^{m}_{m-1,s^{\prime}}\nonumber\\
&+& (\pi^0_{m-1,s}\pi^{m-1}_{m,s}+\pi^0_{m-1,s^{\prime}}\pi^{m-1}_{m,s^{\prime}})\pi^{m}_{m-1,s^{\prime}}\nonumber\\
&+& (\pi^0_{m-1,s}\pi^{m-1}_{m,s^{\prime}}+\pi^0_{m-1,s^{\prime}}\pi^{m-1}_{m,s})\pi^{m}_{m-1,s}\nonumber\\
&=&(\pi^0_{m-1,s}\pi^{m-1}_{m,s}+\pi^0_{m-1,s^{\prime}}\pi^{m-1}_{m,s^{\prime}})(\pi^{m}_{m-1,s}+\pi^{m}_{m-1,s^{\prime}})\nonumber\\
&+& (\pi^0_{m-1,s}\pi^{m-1}_{m,s^{\prime}}+\pi^0_{m-1,s^{\prime}}\pi^{m-1}_{m,s})(\pi^{m}_{m-1,s}+\pi^{m}_{m-1,s^{\prime}})\nonumber\\
&=&\pi^0_{m,s}(\pi^{m}_{m-1,s}+\pi^{m}_{m-1,s^{\prime}})+\pi^0_{m,s^{\prime}}(\pi^{m}_{m-1,s}+\pi^{m}_{m-1,s^{\prime}})\nonumber\\\label{eq:leaf}
\end{eqnarray}
A similar set of derivations for $\sum_{i\in C(v_{r_m})}\sum_{t\in S}J(\phi_{cm_i}=t|v_{r_0})$ yields
\begin{eqnarray}
& &\sum_{i\in C(v_{r_m})}\sum_{t\in S}J(\phi_{cm_i}=t|v_{r_0})\nonumber\\
&=&\sum_{i\in C(v_{r_m})}[(\pi^0_{m,s}\pi^m_{cm_i,s}+\pi^0_{m,s^{\prime}}\pi^m_{cm_i,s^{\prime}})(\pi^{cm_i}_{m,s}+\pi^{cm_i}_{m,s^\prime})\nonumber\\
&+& (\pi^0_{m,s}\pi^m_{cm_i,s^{\prime}}+\pi^0_{m,s^{\prime}}\pi^m_{cm_i,s})(\pi^{cm_i}_{m,s}+\pi^{cm_i}_{m,s^\prime})]\label{eq:leafChild}
\end{eqnarray}
Since for a leaf node $\pi^{cm_i}_{m,s}+\pi^{cm_i}_{m,s^{\prime}}=1$, Equation~\ref{eq:leafChild} can be expressed as
\begin{eqnarray}
\sum_{i\in C(v_{r_m})}\left[\pi^0_{m,s}(\pi^{m}_{cm_i,s}+\pi^{m}_{cm_i,s^{\prime}})+\pi^0_{m,s^{\prime}}(\pi^{m}_{cm_i,s}+\pi^{m}_{cm_i,s^{\prime}})\right]\label{eq:leafChildFinal}
\end{eqnarray}
One may note that for the node $v_{r_m}$, from Equation~\ref{eq:signRelev} we get $\pi^m_{m-1, s}+\pi^m_{m-1, s^{\prime}}+\sum_{i\in C(v_{r_m})}\left[\pi^{m}_{cm_i,s}+\pi^m_{cm_i,s^{\prime}}\right]=1$ and thus adding Equations~\ref{eq:leaf} and ~\ref{eq:leafChildFinal}, we get
\begin{eqnarray}
\sum_{i\in ST_{r_m}}\sum_{t\in S}J(\phi_i=t|v_{r_0})&=& \pi^0_{m,s}+\pi^0_{m,s^{\prime}}\label{eq:genSubTree}
\end{eqnarray}
We collapse the sub-tree rooted at $v_{r_m}$ so as to form a single leaf node and recursively work towards the sub-tree rooted at the child nodes of the root node $v_{r_0}$. Then, for the child node $v_{r_1}$, we have 
\begin{eqnarray}
\sum_{i\in ST_{r_1}}\sum_{t\in S}J(\phi_i=t|v_{r_0})&=& \pi^0_{1,s}+\pi^0_{1,s^{\prime}}
\end{eqnarray}
Taking the sum over all subtrees rooted at the child nodes of $v_{r_0}$, we have
\begin{eqnarray}
\sum_{i\neq r_0}\sum_{t\in S}J(\phi_i=t|v_{r_0})=1,
\end{eqnarray}
which completes our proof that the normalization property holds for the modified softmax function.
\end{proof}
%we further justify that graph soft- max characterizes the real pattern of connectivity distribu- tion precisely by conducting an empirical study in the ex- periment part
In order to show that the modified softmax considers the connectivity pattern of a node, we prove next how the modified softmax score of the root node changes with respect to the shortest path length to any node as followed by Wang et al. \cite{wang2018graphgan} in unsigned networks.

%considers the graph structure in generating the probability scores.
\begin{theorem}
In the modified softmax $J(\phi_i|v_j, \theta_J)$ decreases with increasing shortest path distance between $v_i$ and $v_j$.
\end{theorem}
\begin{proof}
Since BFS tree preserves the shortest path from the root node to the other nodes, so for a root node $v_{r_0}$, the probability $J(\phi_{r_n}|v_j, \theta_J)$ for an arbitrary node $v_{r_n}$ that is $n$ hops away is proportional to the product of $n$ exponential terms (the sign-specific relevance probabilities) and thus, decreases exponentially.
\end{proof}

Therefore, on the basis of proof of Theorem 2, we conclude that the modified softmax can capture the connectivity pattern of a node as it always follows the shortest path distance. We follow Wang et al. \cite{wang2018graphgan} who discuss that capturing the connectivity pattern inherently captures the graph structure and therefore, conclude that SigGAN can capture the graph structure of signed networks. We next show that the modified softmax improves upon the computational efficiency in deriving $J(\phi_i|v_j)$.

%Following Theorem 2, we further justify that graph soft- max characterizes the real pattern of connectivity distribu- tion precisely by conducting an empirical study in the ex- periment part.

\begin{theorem}
For large networks with power-law degree distributions, to derive the probability $J(\phi_i|v_j)$, the average number of node representation updates is $O(d\ln |V|)$, where $d$ is the average degree of the nodes. 
\end{theorem}
\begin{proof}
For an arbitrary node $v_i$, calculating $J(\phi_i|v_j)$ requires updating the representations of the nodes that lie in the path from $v_j$ to $v_i$ and also those that are directly connected to these. As derived in ~\cite{fronczakPRE04}, for networks with power-law degree distribution, given as, $P(k)\sim k^{-\alpha}$, in the limit $|V|\rightarrow \infty$, the average path length between the nodes depends on the scaling exponent $\alpha$. When $\alpha<3$, the average path length saturates to $\frac{2}{3-\alpha}+\frac{1}{2}$, whereas for $\alpha=3$, the average path length approximates to $\frac{\ln |V|}{\ln\ln |V|}$ and when $\alpha>3$, then the same tends to $\ln |V|$. Thus, as the average path length for a large class of networks is maximally $O(\ln |V|)$, hence the total number of nodes whose representations would be updated is $O(d\ln |V|)$ on an average.
\end{proof}
We next intuitively discuss whether the modified softmax can maintain the structural balance property if that exists in the network. For any node triplet, $v_{r_i}, v_{r_j}$ and $v_{r_k}$, if edges exist between $(v_{r_i}, v_{r_j})$ and $(v_{r_j}, v_{r_k})$, then from Equation~\ref{eq:genRoot}, we find that $J(\phi_{r_k}=s|v_{r_i})$ contains the expression $(\pi^{i}_{j,s}\pi^{j}_{k,s}+\pi^{i}_{j,s'}\pi^{j}_{k,s'})$ indicating that the corresponding probability of an edge between $v_{r_i}$ and $v_{r_k}$ being positive is high when either the probability of both edges $(v_{r_i}, v_{r_j})$ and $(v_{r_j}, v_{r_k})$ being positive is high or both of them being negative is high. This relation would be learnt from the discriminator function which would return true with a very high probability when the sign of edge $(v_{r_i}, v_{r_k})$ is determined as positive by the generator. Same holds for the negative edges between $(v_{r_i}$ and $v_{r_k})$. If either $\pi^{i}_{j,s}$  and $\pi^{j}_{k,s'}$  or $\pi^{i}_{j,s'}$  and $\pi^{j}_{k,s}$ are both high, then the generator would automatically increase the probability of $J(\phi_{r_k}=s'|v_{r_i})$. 
We next investigate the performance of the SigGAN approach.
%\subsection{Implementation}
%\textcolor{red}{Roshni please write}

\section{Experiments} \label{s:expt}
%%\subsection{Experiments}
In this Section, we provide details of our experiments. We initially discuss the datasets used followed by the existing research works that we use for comparisons and finally, discuss the results obtained.

\subsection{Dataset Details}\label{signdata}
We analyze the performance of SigGAN on five datasets, \textit{Slashdot}, \textit{Epinions}, \textit{Bitcoin}, \textit{Reddit} and \textit{Wiki-RFA}. A brief overview of the datasets is as follows :

\begin{itemize}
    \item \textit{Slashdot}: This dataset\footnote{https://snap.stanford.edu/data/soc-sign-Slashdot090221.html} is created from the relationships among users who tag each other as either friends or foes from the news website, \textit{Slashdot}~\cite{leskovec2010signed}.
    \item \textit{Epinions}: This dataset is created from the consumer reviewer site, \textit{Epinions}, where users share trust or distrust relationships among themselves\footnote{https://snap.stanford.edu/data/soc-sign-epinions.html}~\cite{leskovec2010signed}.
    \item \textit{Reddit}: We obtained this dataset from the SNAP database\footnote{https://snap.stanford.edu/data/soc-RedditHyperlinks.html}. This is a directed network of hyperlinks created from subreddit pairs created from posts that create hyperlink from one post to another~\cite{kumar2018community}. The sign of the link is based on the sentiment of the source subreddit towards the hyperlinked subreddit. 
\item \textit{Bitcoin}: This dataset represents the trust or distrust relationship between users who trade on the Bitcoin OTC\footnote{https://snap.stanford.edu/data/soc-sign-bitcoin-otc.html}~\cite{kumar2016edge} platform. 
\item \textit{Wiki-RFA}: This dataset\footnote{https://snap.stanford.edu/data/wiki-RfA.html} represents the voting relationship among the users, i.e., nodes represent users and the edges represent the votes among the users~\cite{west2014exploiting}. Therefore, the sign of the edges represent supporting, neutral, or opposing vote.
\end{itemize}

\begin{table}{}
\centering
\begin{tabular}{|c|c|c|c|c|c|c|}
\hline\hline
Dataset&\textit{Slashdot}&\textit{Epinions}&\textit{Reddit}&\textit{Bitcoin}&\textit{Wiki-RFA}\\
\hline\hline
Nodes&7000&7000&6999&5877&7118\\
\hline			
Edges&431098&734408&277050&21436&201386\\
\hline
Positive Edges&321030&642397&249580&18282&157206\\
\hline
Negative Edges&110068&92011&27470&3154&44180\\
\hline
\end{tabular}
\caption{Dataset Details}\label{tab:sgndata2}
\end{table}
%the volume of the data is enormously large and, hence we follow Shen et al.~\cite{shen2018deep} to create graphs of each of the dataset by selecting the the top $7000$ nodes with the highest degree and their corresponding edges. 

For our experiments, we considered the dataset for \textit{Wiki-RFA}, \textit{Reddit} and \textit{Bitcoin} and  follow Shen et al.~\cite{shen2018deep} to create graphs for \textit{Slashdot} and \textit{Epinions} datasets. We outline the details of each of these datasets in the Table~\ref{tab:sgndata2}. Our observations indicate that the number of negative edges is quite low compared to the number of positive edges and the sparsity of the graphs differ across datasets.
\subsection{Baselines}\label{s:signbase}
We compare SigGAN with 3 state-of-the-art deep learning-based techniques for signed link prediction. Each of these techniques generates node representations using different strategies that are used to determine the sign of the links. We briefly discuss these baselines next.

\begin{table*}{}
\centering
\begin{tabular}{|c|c|c|c|c|c|c|}
\hline\hline
Dataset&$Slashdot$&$Epinions$&$Reddit$&$Bitcoin$&$Wiki-RFA$\\
\hline\hline
$L1$&0.74&0.883&0.93&0.86&0.76\\
\hline			
$L2$&0.76&0.878&0.94&\textbf{0.87}&0.76\\
\hline
$Had$&\textbf{0.79}&\textbf{0.89}&\textbf{0.96}&0.86&0.75\\
\hline
$Avg$&0.78&0.883&0.93&0.868&\textbf{0.77}\\
\hline
$Concat$&0.74&0.87&0.92&0.87&0.74\\
\hline
\end{tabular}
\caption{Mean micro F1-score of the SigGAN for different edge vector representations}\label{tab:sgnvec1}
\end{table*}

\begin{enumerate}
\item {\textit{DNE-SBP}}:  Shen et al.~\cite{shen2018deep} proposed a deep neural embedding with structural balance preservation (DNE-SBP) model to learn graph representations for signed networks. The method uses a semi-supervised stacked auto-encoder to explore the connectivity patterns of a node. Further, the SigGAN exploited this connectivity relationship of a node and extended structural balance theory to generate node embeddings. They explicitly ensured that more importance is given in reconstructing the negative edges than the positive edges. 
%We refer to this baseline as $B_1$ in table~\ref{tab:sgnres6}.

\item{\textit{SIDE}}: 
Kim et al.~\cite{kim2018side} proposed a method for representation learning in signed directed networks (SIDE) that incorporates sign, direction, and proximity relationships of the nodes to generate low dimensional vector representations of these nodes. This model generates random walks across the graph that captures the proximity between nodes in the random walks, and subsequently generates the likelihood of co-occurrence of a node pair based on this information. Further, in the next step, using this information, SigGAN generates embedding of the nodes following a skip-gram technique with negative sampling. 
%This approach is referred to by $B_2$ in Table~\ref{tab:sgnres6}.

\item{\textit{SNE}}: Yuan et al.~\cite{yuan2017sne} proposed a technique for signed network embedding (SNE) that uses a  log-bilinear model for generating the node embeddings. The model considers the node representations of all the nodes along different paths of a node to generate its embedding. The model computes the embedding of the target node given the information of the predecessors along the path and the signed relationship among these predecessors by maximizing the log-likelihood. 
%We refer to this approach by $B_3$ in Table~\ref{tab:sgnres6}.
\end{enumerate}

For our experiments, we perform stochastic gradient descent to update parameters in SigGAN with learning rate $0.001$. For each iteration, we set the number of positive samples  as $20$ and repeat this for $10$ epochs for both generator and discriminator respectively with batch size as $32$. Subsequently, we set the number of epochs for SigGAN as $10$ and consider the dimension of node embedding as $50$. We set the values of the hyper-parameter by cross-validation which we discuss in details in Section \ref{s:hyp}. We used the available codes provided by the authors of each of the baselines for our experiments. We next outline the performance measures used for our investigations, the results obtained and the impact of the values of hyper-parameters on the performance of SigGAN.

\begin{table*}{}
\centering
\begin{tabular}{|c|c|c|c|c|c|c|}
\hline\hline
Dataset&$Slashdot$&$Epinions$&$Reddit$&$Bitcoin$&$Wiki-RFA$\\
\hline\hline
$SigGAN$&\textbf{0.79}&\textbf{0.89}&\textbf{0.96}&\textbf{0.86}&\textbf{0.75}\\
\hline			
$DNE-SBP$&0.68&0.83&0.86&0.83&0.70\\
\hline
$SIDE$&0.73&0.87&0.93&0.84&0.70\\
\hline
$SNE$&0.73&0.86&0.93&0.84&0.69\\
\hline
\end{tabular}
\caption{Mean micro F1- score of the SigGAN with the existing research works}\label{tab:sgnres4}
\end{table*}

\section{Results and Discussions} \label{s:res}
In this Section, we compare the performance of the SigGAN with the baselines. As there is a class imbalance in signed networks, i.e., the number of negative edges is very less compared to the number of positive edges (as shown in Table \ref{tab:sgndata2}), we use \textit{micro-F1} score which captures the effectiveness of SigGAN in predicting both the signs. The \textit{micro-F1} score is calculated as the harmonic mean of average \textit{precision} (average of the respective precision of SigGAN in predicting positive and negative edges respectively) score and average \textit{recall} (average of the recall of the SigGAN in predicting positive and negative edges respectively) score. Therefore, a higher \textit{micro-F1} score of SigGAN in link prediction implies that SigGAN can predict the sign of both positive and negative edges effectively. We also evaluate the performance of SigGAN in handling the sparsity of the graphs by inducing different levels of sparsity in the datasets. Additionally, we perform a case study where we investigate whether \textit{structural balance theory} property is satisfied in the node embedding generated by SigGAN.

\subsection{Link Prediction} \label{s:sgnlink}
Link prediction in signed networks is to identify the sign of an edge. We compare the performance of SigGAN with the baselines, DNE-SBP, SIDE and SNE in identifying the sign of the edge. Based on the experiments followed in the existing research works, we use a logistic regression classifier~\cite{derr2018signed, shen2018deep} that is trained on the vector representation of the edges as features. Since SigGAN and the baselines generate the node vector representations, we initially derive the corresponding edge vector representation, $e_{ij}$. We follow $5$ different methods, L1-norm, L2-norm, Hadamard product, average and concatenation to determine $e_{ij}$ as followed in existing research works~\cite{kim2018side, shen2018deep}. We, then, train the \textit{logistic regression} model on the training set for the edge representation. We use the trained \textit{logistic regression} model to predict the sign of the links in the validation test. We use 5-fold cross-validation as shown in~\cite{kim2018side} and report the mean \textit{micro F1-score} which is calculated as the average of \textit{micro F1-score} over $5$ independent runs. A higher mean \textit{micro F1-score} indicates better performance in predicting the sign of both the positive and negative links.  We repeat the experiments for \textit{Slashdot}, \textit{Epinions}, \textit{Wiki-RFA}, \textit{Reddit} and \textit{Bitcoin}.

\begin{table*}{}
\centering
\begin{tabular}{|c|c|c|c|c|c|c|c|c|c|c|c|c|c|c|c|c|c|c|}
\hline\hline
{Appr}&\multicolumn{4}{c|}{Slashdot}&\multicolumn{4}{c|}{Epinions}&\multicolumn{4}{c|}{Wiki-RFA}\\
\hline
\cline{2-5}
    & {80\%} & {60\%} & {40\%} & {20\%} & {80\%} & {60\%} & {40\%} & {20\%}& {80\%} & {60\%} & {40\%} & {20\%}\\
\hline
$SigGAN$&0.72&\textbf{0.76}&\textbf{0.77}&\textbf{0.78}&\textbf{0.89}&\textbf{0.86}&\textbf{0.87}&\textbf{0.88}&\textbf{0.76}&0.73&\textbf{0.76}&\textbf{0.74}\\
\hline
$DNE-SBP$&0.71&0.76&0.76&0.78&0.89&0.84&0.86&0.84&0.72&0.73&0.73&0.71\\
\hline
%$SGCN$&xxx&yyy&xxx&yyy&xxx&xxx&xxx&xxx&xxx&xxx&xxx&xxx\\
%\hline	
$SIDE$&0.71&0.76&0.76&0.75&0.87&0.84&0.86&0.87&0.74&\textbf{0.75}&0.73&0.71\\
\hline	
$SNE$&\textbf{0.75}&0.75&0.77&0.75&0.76&0.76&0.86&0.84&0.74&0.73&0.71&0.71\\
\hline
\end{tabular}
\caption{Comparison results of mean micro F1 score of SigGAN with the baselines for different different levels of sparsity}\label{tab:sgnres6}
\end{table*}

\begin{table*}{}
\centering
\begin{tabular}{|c|c|c|c|c|c|c|c|c|c|c|c|c|c|c|c|c|c|c|}
\hline\hline
{Appr}&\multicolumn{4}{c|}{Bitcoin}&\multicolumn{4}{c|}{Reddit}\\
\hline
\cline{2-5}
    & {80\%} & {60\%} & {40\%} & {20\%} & {80\%} & {60\%} & {40\%} & {20\%}\\
\hline
$SigGAN$&\textbf{0.81}&\textbf{0.83}&\textbf{0.84}&\textbf{0.86}&0.90&\textbf{0.92}&\textbf{0.93}&\textbf{0.94}\\
\hline
$DNE-SBP$&0.79&0.81&0.81&0.82&0.86&0.85&0.85&0.85\\
\hline
%$SGCN$&xxx&yyy&xxx&yyy&xxx&xxx&xxx&xxx&xxx&xxx&xxx&xxx\\
%\hline	
$SIDE$&0.80&0.81&0.81&0.83&\textbf{0.91}&0.89&0.90&0.92\\
\hline	
$SNE$&0.81&0.81&0.82&0.82&0.90&0.90&0.90&0.91\\
\hline
\end{tabular}
\caption{Comparison results of mean micro F1 score of the SigGAN with the baselines for different levels of sparsity}\label{tab:sgnres16}
\end{table*}

\begin{figure*}
     %\centering
     \begin{subfigure}[b]{0.33\textwidth}
    \includegraphics[width=2.1in]{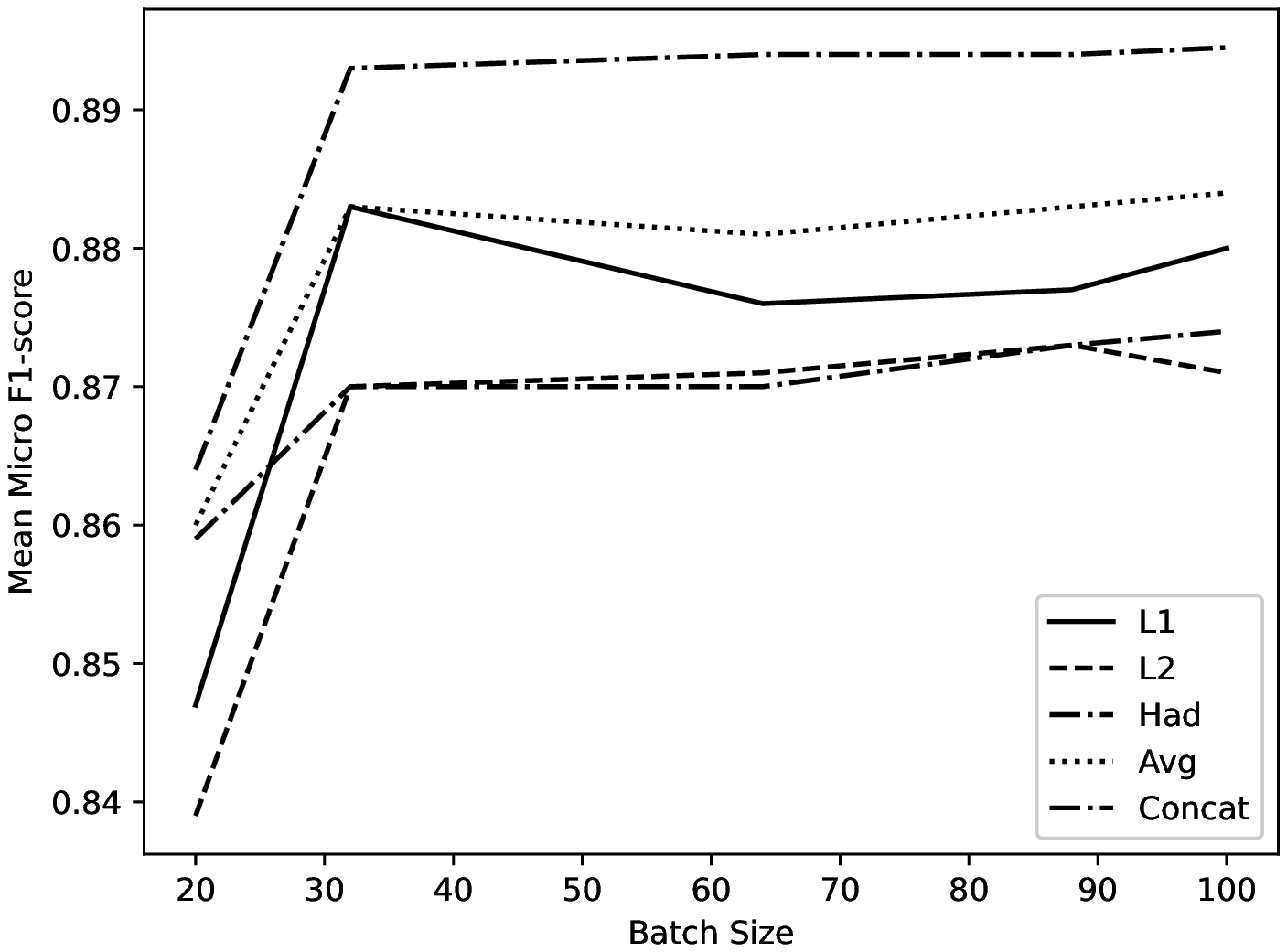}
        \caption{}
        \label{fig:p1}
     \end{subfigure}
    \hfill
     \begin{subfigure}[b]{0.33\textwidth}
    \includegraphics[width=2.1in]{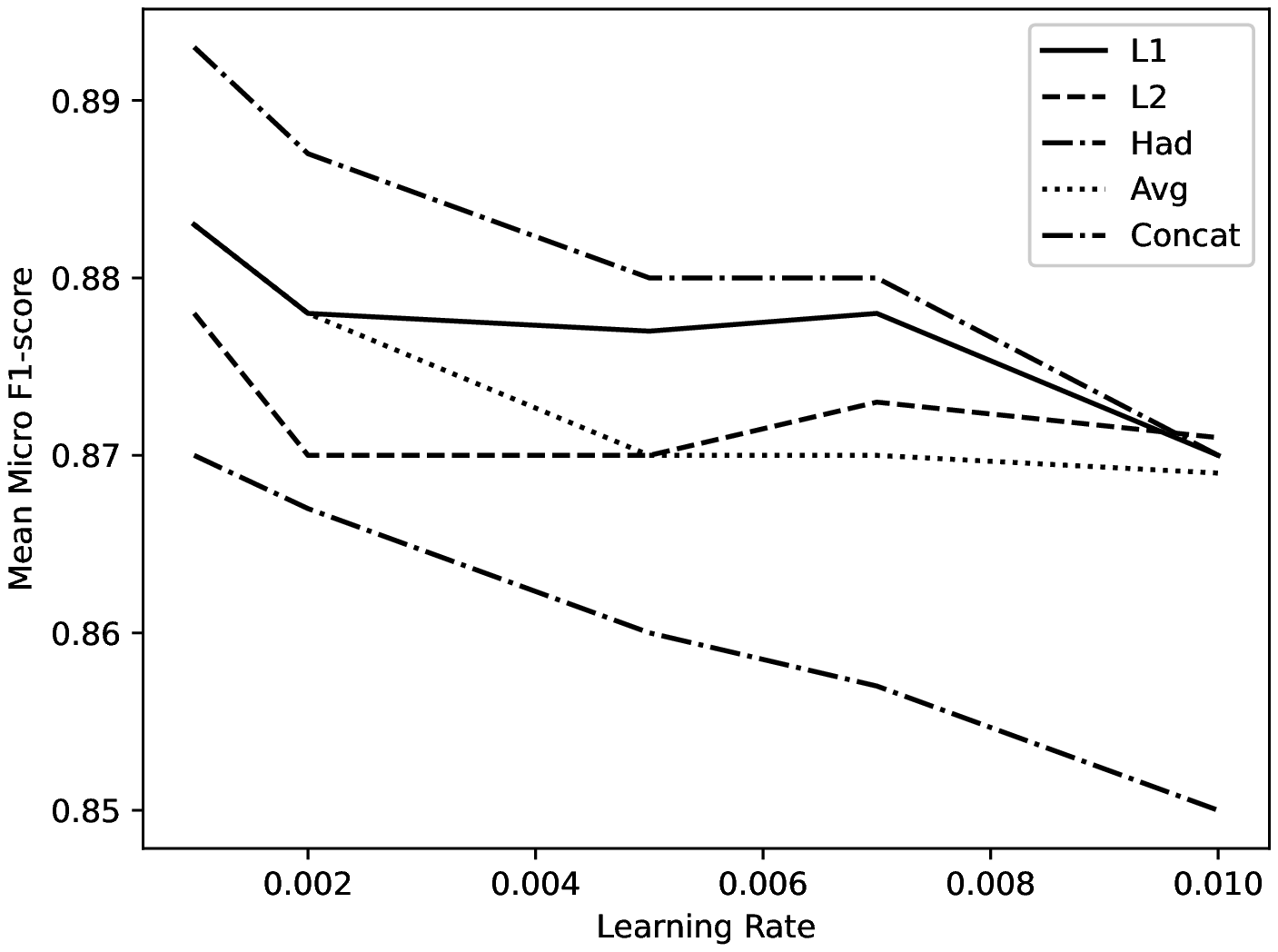}
        \caption{}
        \label{fig:p4}
     \end{subfigure}
     \begin{subfigure}[b]{0.33\textwidth}
    \includegraphics[width=2.1in]{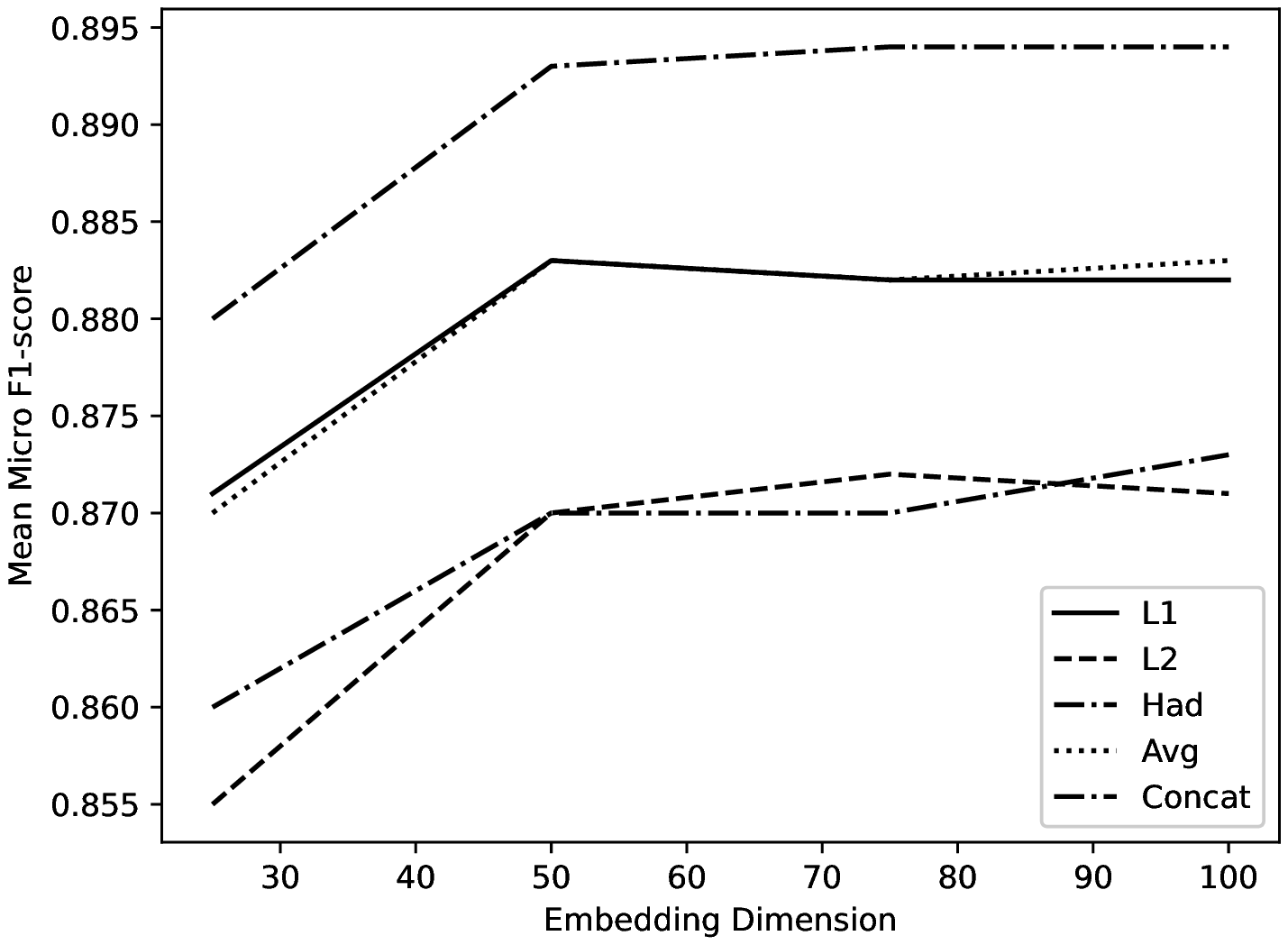}
        \caption{}
        \label{fig:p5}
     \end{subfigure}
     \caption{Comparing the effect of change in batch size, learning rate and embedding dimension on Mean Micro F1-score for \textit{Epinions} in Figure \ref{fig:p1}, Figure \ref{fig:p4} and Figure \ref{fig:p5} respectively.}
      \label{fig:pa1}
\end{figure*}

\begin{figure*}
     \centering
\begin{subfigure}[b]{0.45\textwidth}
    \includegraphics[width=2.5in]{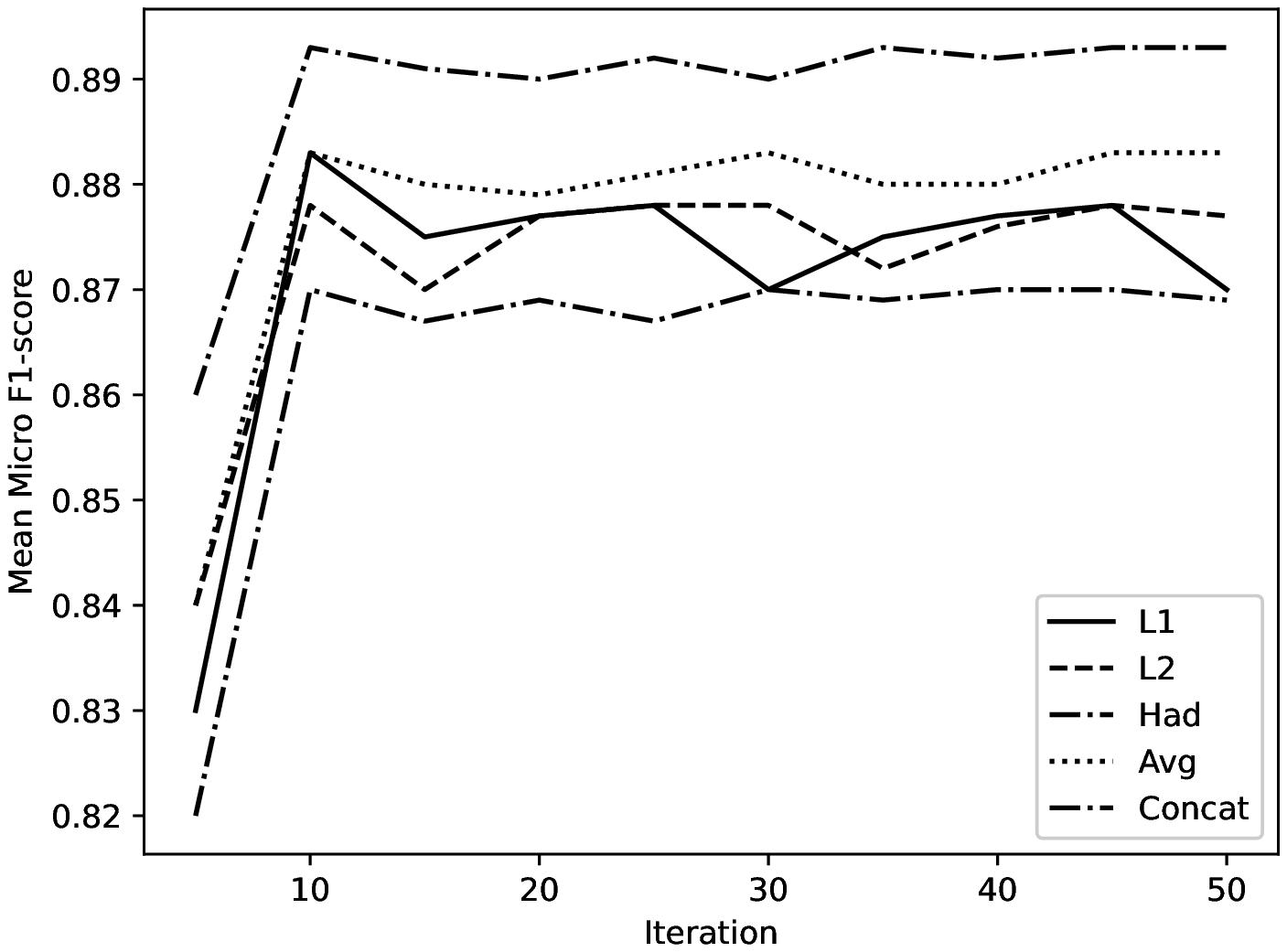}
        \caption{}
        \label{fig:p2}
     \end{subfigure}
     \begin{subfigure}[b]{0.45\textwidth}
    \includegraphics[width=2.5in]{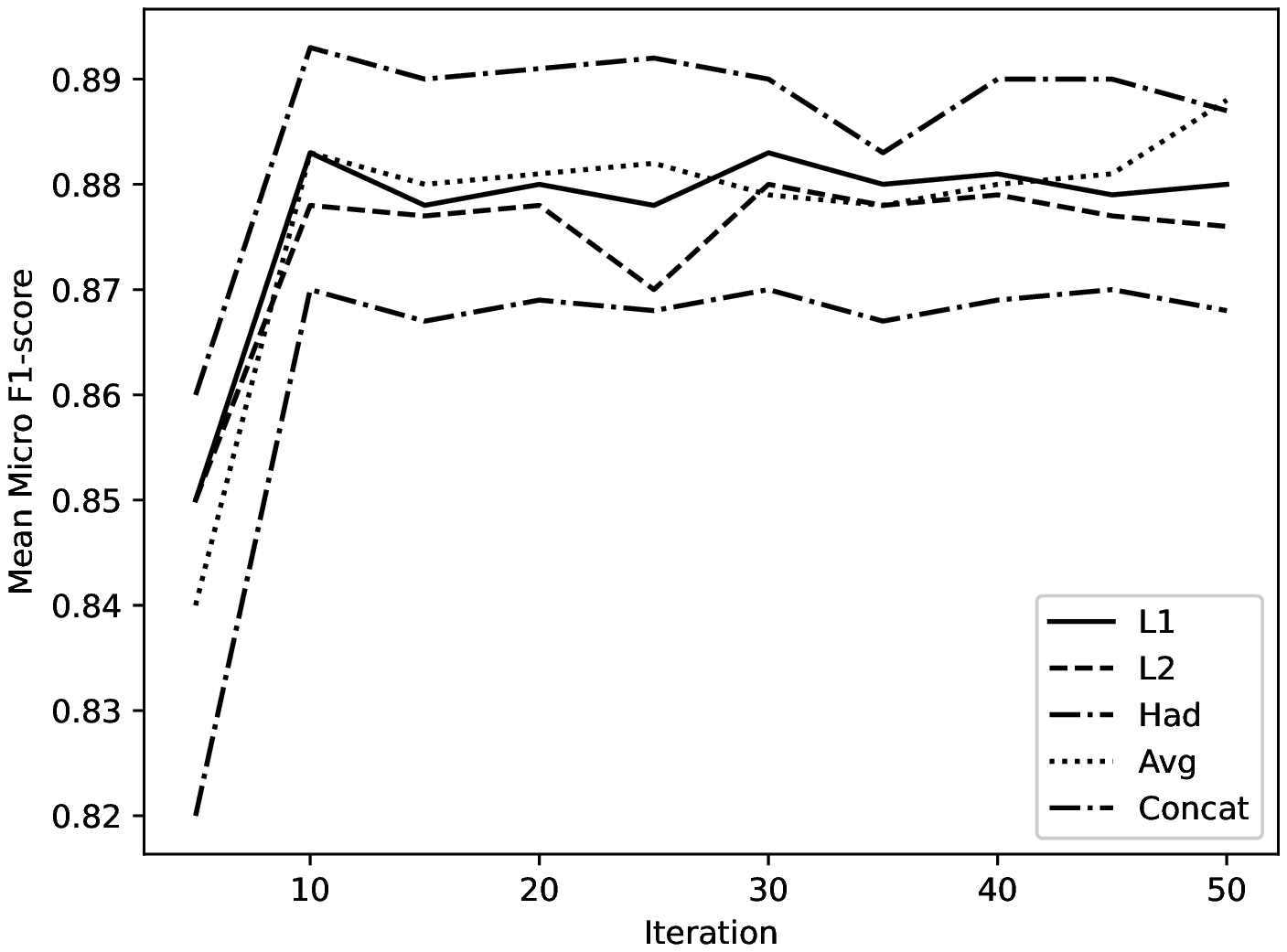}
        \caption{}
        \label{fig:p3}
     \end{subfigure}
      \caption{Comparing the effect of change of number of epochs for Generator and Discriminator and number of epochs of SigGAN on Mean Micro F1-score for \textit{Epinions} dataset in Figure \ref{fig:p2} and Figure \ref{fig:p3} respectively.}
      \label{fig:pa2}
\end{figure*}
\par We show the mean micro F1-score for different types of edge representations for SigGAN in Table \ref{tab:sgnvec1}. We observe that the edge representation by Hadamard product performs the best for \textit{Slashdot}, \textit{Epinions} and \textit{Reddit}, edge representation by average for \textit{Wiki-RFA} and edge representation by L2-norm for \textit{Bitcoin} respectively. Furthermore, the edge representation by Hadamard product for \textit{Wiki-RFA} and \textit{Bitcoin} is very near to the best performance with a difference of $0.02$ and $0.01$ respectively. Based on these observations, we choose edge representation by Hadamard product for SigGAN. In~\cite{kim2018side, shen2018deep}, it has been shown that the baselines have better performance when the logistic regression model is trained by the edge vector representation by Hadamard product than the other representations. Therefore, in order to maintain uniformity, we consider the edge vector representation by Hadamard product for our comparison results. We show the mean micro F1-score in Table \ref{tab:sgnres4}. Our observations from Table \ref{tab:sgnres4} indicates the SigGAN provides better performance than all the existing research works in all the datasets. The mean micro F1-score of the SigGAN is around $3-11\%$ more than the existing research works. The performance of the SigGAN is the best for \textit{Reddit} dataset with around $0.96$ followed by $0.89$ in \textit{Epinions}. We next discuss the values of each hyper-parameter and it's effect on the performance of SigGAN.

\subsubsection{Hyperparameter Analysis} \label{s:hyp}: We discuss next how we set the values of the different parameter or hyperparameter for SigGAN. For each parameter or hyperparameter, we vary its values and calculate the mean micro F1-score for $5$ independent runs for each value. Therefore, we select that value for a parameter or hyperparameter which gives the best mean micro F1-score during training. We consider edge representations by L1-norm, L2-norm, Hadamard product, average and concatenation. We show our observations for \textit{Epinions} in Figure \ref{fig:pa1} and Figure \ref{fig:pa2} respectively. Based on our observations from Figure \ref{fig:p1}, we find the value of batch-size of $32$ provides the maximum mean micro F1-score irrespective of the function for edge representation except for L2-norm and Hadamard product when batch size is $100$. However, our observations indicate that the increase in results for L2-norm and Hadamard product when batch size is $100$ is very small and further, it might lead to overfitting if we increase the batch size to $100$. Therefore, we consider $32$ as batch size for SigGAN. Similarly, we observe for learning rate as $0.001$, SigGAN has the maximum mean micro F1-score. Additionally, to decide the embedding dimension, we vary the dimension from $25-100$ and calculate the result of SigGAN. Our observations as shown in Figure \ref{fig:p5} indicate that increasing the dimension from $50$ to $100$ provides almost similar mean micro F1-score and dimension of size $25$ has the least mean micro F1-score. As we want to have as low dimension as possible, we choose $50$ as the node embedding dimension. We specify the number of epochs by which we train SigGAN as $s$ and $t$ as the number of epochs for which discriminator and generator runs for each value of $s$. Our observations from Figure \ref{fig:p2} and Figure \ref{fig:p3} indicate that SigGAN has the maximum mean micro F1-score irrespective of the function for edge representation when both $s$ and $t$ are $10$ and further, the increase in the number of epochs has very less impact in the mean micro F1-score. Therefore, the modified Softmax specifically for signed networks and simultaneous training of SigGAN by equal number of positive and negative edges leads to quicker convergence as observed in Figure \ref{fig:p2} and Figure \ref{fig:p3}. Therefore, based on our training mean micro F1-score, we set these values for all the parameter and hyperparameter values.
%\textbf{window size left}

\subsection{Handling Sparsity}\label{s:sparse}
We investigate the performance of SigGAN in handling sparse graphs.  For our experiments, we induce sparsity in the original graph by randomly removing $20\%$, $40\%$, $60\%$ and $80\%$ of the edges from the original graph. For each of these sparse graphs, we generate the embedding of each node using SigGAN. Further, to create a graph with specific sparsity, say $20\%$, we repeat the procedure $5$ times so that we have $5$ different graphs with $20\%$ sparsity. For each of these graphs, we then evaluate the performance of SigGAN in predicting the sign of the link by following the procedure discussed in Section \ref{s:sgnlink}. Additionally, since we follow edge representation by Hadamard product in Section \ref{s:sgnlink}, we report only the results for  edge representation by Hadamard product here. We repeat the same experiment for \textit{Slashdot}, \textit{Epinions}, \textit{Wiki-RFA}, \textit{Reddit} and \textit{Bitcoin}. We calculate the mean \textit{micro F1-score} as the average of \textit{micro F1-score} over $5$ different graphs with same sparsity and \textit{micro F1-score} over a graph is calculated by 5-fold cross validation. We show our observations in Table \ref{tab:sgnres6} and Table \ref{tab:sgnres16} respectively. We discover that SigGAN ensures better performance than the baselines in all the datasets when the sparsity is increased except for when $80\%$ of the edges were removed for \textit{Slashdot}, \textit{Reddit} and $40\%$ for \textit{Wiki-RFA}. Subsequently, we observe that SigGAN ensures almost consistent performance irrespective of the sparsity in the signed network, i.e., $0.72-0.78$ for \textit{Slashdot}, $0.86-0.89$ for \textit{Epinions}, $0.73-0.76$ for \textit{Wiki-RFA}, $0.81-0.86$ for \textit{Epinions} and $0.90-0.94$ for \textit{Reddit}.

\subsection{Case Study : Analysis of Embedding by SigGAN} \label{signcase}
We now study whether the network representations learned by SigGAN follows the extended structural balance property, i.e., positively connected nodes are closer in the representation space than the negatively connected nodes, as discussed previously in ~\cite{islam2018signet}. For our experiments, we follow the procedure adopted in existing research works~\cite{islam2018signet,kim2018side}. In order to measure this, we calculate the average positive edge distance, $APED$ as the average distance between the vector representation of the positively connected nodes and $ANED$ for negatively connected nodes, i.e.,

\begin{eqnarray}
APED= \frac{(\sum_{i=1}^{n} \sum_{j=1}^{n} A_{ij}^+ d_{ij})}{(\sum_{i=1}^{n} A_{ij}^+)} \\
ANED= \frac{(\sum_{i=1}^{n} \sum_{j=1}^{n} A_{ij}^- d_{ij})}{(\sum_{i=1}^{n} A_{ij}^-)} 
\end{eqnarray}

where, $A_{ij}^+$ is $1$ if node $i$ and node $j$ are positively connected else $0$, $A_{ij}^-$ is is $1$ if node $i$ and node $j$ are negatively connected else $0$ and $d_{ij}$ is the Euclidean distance between the vector representations of $i$ and $j$. Therefore, for a network to follow \textit{extended structural balance theory}, $ANED$ should be greater than $APED$. We randomly select equal number of positively and negatively connected edges, i.e., $40\%$ of the total negatively connected edges. We repeat the same experiment for all the datasets, i.e., 
\textit{Slashdot}, \textit{Epinions}, \textit{Reddit}, \textit{Bitcoin} and \textit{Wiki-RFA}. The results are shown in Table~\ref{tab:sgnres5} which indicates that $APED$, is around $1.5-4.0$ less than $ANED$, irrespective of the dataset. We discover SigGAN has the best performance for \textit{Epinions} and \textit{Bitcoin} where the average distance between positively and negatively connected nodes are the highest. Our observations confirm that SigGAN can ensure that the positively connected nodes have a smaller distance than the negatively connected nodes. Thus, these results indicate that for signed networks, the SigGAN ensures positively connected node pairs are placed closer than the negatively connected node pairs and can, therefore, ensure extended structural balance theory property.

%We calculate the distance between the vector representation of a pair of nodes when connected positively or negatively. We repeat this experiment for an equal number of positive and negative edges selected randomly and calculate the mean distance between a pair of nodes when connected positively or negatively to ensure our observations are not biased. We repeat the same experiment for all the datasets, i.e., \textit{Slashdot}, \textit{Epinions}, \textit{Reddit}, \textit{Bitcoin} and \textit{Wiki-RFA}. The results are shown in Table~\ref{tab:sgnres5} which indicates that the distance between positively connected node pair is around $1.5-4.0$ less than the negatively connected node pair irrespective of the dataset. We observed that the difference between positively and negatively connected nodes as maintained by the SigGAN is similar to the existing works as reported by Islam et al.~\cite{islam2018signet} and Kim et al.~\cite{islam2018signet,kim2018side} where the distance varied between $0.4-0.679$ and $1.5-1.8$ for different datasets. We discover that the average distance for \textit{Epinions} and \textit{Bitcoin} are higher than the rest. Thus, these results indicate that for signed networks, the SigGAN ensures positively connected node pairs are placed closer than the negatively connected node pairs and can, therefore, ensure \textit{structural balance theory} property. 

\begin{table*}{}
\centering
\begin{tabular}{|c|c|c|c|c|c|c|}
\hline\hline
Dataset&$Slashdot$&$Epinions$&$Reddit$&$Bitcoin$&$Wiki-RFA$\\
\hline\hline
$APED$&73&54.5&10.90&45&10\\
\hline			
$ANED$&74.50&58.56&12&48&11\\
\hline
\end{tabular}
\caption{We show $APED$ and $ANED$ for \textit{Slashdot}, \textit{Epinions}, \textit{Reddit}, \textit{Bitcoin} and \textit{Wiki-RFA}}\label{tab:sgnres5}
\end{table*}

%%\textcolor{red}{I am not convinced how is the balanc

\begin{comment}

\begin{table*}{}
\centering
\begin{tabular}{|c|c|c|c|c|c|c|c|c|c|c|c|c|c|c|c|c|}
\hline\hline
\multirow{}{Algorithms}&\multicolumn{4}{c|}{Reddit}&\multicolumn{4}{c|}{Bitcoin}&\multicolumn{4}{c|}{Wiki-RFA}\\
\hline\hline
\cline{2-5}
    & {20\%} & {40\%} & {60\%} & {80\%} & {20\%} & {40\%} & {60\%} & {80\%}& {20\%} & {40\%} & {60\%} & {80\%}\\
\hline
$Proposed$&xxx&yyy&xxx&yyy&xxx&xxx&xxx&xxx&xxx&0.71&0.76&0.74\\
\hline
$DNE-SBP$&xxx&yyy&xxx&yyy&xxx&xxx&xxx&xxx&0.72&0.73&0.73&0.71\\
\hline
$SGCN$&xxx&yyy&xxx&yyy&xxx&xxx&xxx&xxx&xxx&yyy&xxx&yyy\\
\hline	
$SIDE$&xxx&yyy&xxx&yyy&xxx&xxx&xxx&xxx&xxx&yyy&xxx&yyy\\
\hline	
$SNE$&0.75&0.75&0.77&0.75&0.76&0.76&0.86&0.84&0.74&0.73&0.71&0.71\\
\hline
\end{tabular}
\caption{Comparison results of micro F1-score  of the SigGAN with the baselines is shown for \textit{Reddit}, \textit{Bitcoin} and \textit{Wiki-RFA} datasets.}\label{tab:sgnres2}
\end{table*}
%\section{References}

\begin{table*}{}
\centering
\begin{tabular}{|c|c|c|c|c|c|c|}
\hline\hline
Dataset&$Slashdot$&$Epinions$&$Reddit$&$Bitcoin$&$Wiki-RFA$\\
\hline\hline
$Positive$&73&54.5&10.90&45&10\\
\hline			
$Negative$&74.50&58.56&12&48&11\\
\hline
\end{tabular}
\caption{Mean distance between positively(or, negatively) connected nodes by the SigGAN}\label{tab:sgnres5}
\end{table*}
\end{comment}

\subsection{Summary of Insights}~\label{s:sum}
We discuss SigGAN in comparison to GraphGAN \cite{wang2018graphgan} and it's applicability to signed networks followed by an analysis of the performance of the SigGAN with respect to the baselines followed by the shortcomings of SigGAN next. We develop SigGAN based on GraphGAN \cite{wang2018graphgan} which utilizes the inherent characteristics of the graph structure in a GAN based framework to generate node embedding. Although GraphGAN \cite{wang2018graphgan} is very effective for link prediction in unsigned networks, it can not be directly applied to signed networks. Therefore, in this paper, we propose SigGAN which is not a direct extension of GraphGAN but explicitly considers the specific characteristics of signed networks, like \textit{structural balance theory} and \textit{high imbalance in number of positive and negative edges}. We propose a modified softmax for signed networks which integrates \textit{structural balance theory} with the existing properties from unsigned networks, like \textit{normalization} and \textit{graph structure awareness}. Therefore, by integration of these properties in the proposed modified softmax, we ensure that the generator is highly effective in generating negative samples for the discriminator. We show in Subsection \ref{s:modSmax} how we formulate modified softmax such that it considers the properties. We prove appropriate theorems in Subsection \ref{s:modSmax} to show that \textit{normalization}, \textit{graph structure awareness} and \textit{computational efficiency}, are satisfied in the generator. Furthermore, to show that \textit{structural balance theory} is satisfied by SigGAN, we prove through intuition in Subsection \ref{s:modSmax} and through a case study in Subsection \ref{signcase}. In order to handle \textit{high imbalance in number of positive and negative edges} and high information content in the negative edges, we train the discriminator with equal number of positive and negative edges. This further ensures that SigGAN can ensure prediction of negative edge effectively. Our observations in Subsection \ref{s:sgnlink} and Subsection \ref{s:sparse} shows that SigGAN has a higher mean micro F1-score than the baselines. As the micro F1-score considers the efficiency in predicting both the positive and the negative edges, Table \ref{tab:sgnres4} - \ref{tab:sgnres16} indicates the high performance of SigGAN in capturing both the positive and negative edges for original and sparse graphs respectively. Therefore, we believe the design of the modified softmax specific for signed networks for generator and simultaneous training by both positive and negative edges for discriminator are the main reasons behind success of SigGAN. While the generator is effective in generating negative samples irrespective of the sign of the edge considering the signed network characteristics, the discriminator learns to handle the imbalance in data.

\subsubsection{Limitations of SigGAN} : Although SigGAN ensures high effectiveness in comparison with the existing baselines irrespective of the dataset and presence of sparsity, we observe few limitations in SigGAN which we discuss next. A critical analysis of SigGAN would be that we did not explore the applicability of different discriminator models which can improve the performance. We evaluate the execution time for SigGAN and the other baselines on a machine with an Intel Core $i7$ $4.2GHz$ CPU and $32$ GB RAM. We observe during sparsity analysis that SigGAN is linearly scalable when the fraction of edges is increased gradually from $20\%$ of the total edges to all the edges irrespective of the dataset. The time taken are $32$, $40$, $51$, $60$ and $72$ minutes respectively for $20\%$, $40\%$, $60\%$, $80\%$ and $100\%$ of the edges for \textit{Epinions}. Furthermore, on comparing the execution time of SigGAN with the baselines, we observe that the execution time of DNE-SBP is the best followed by SigGAN and SIDE. For example, the execution time on \textit{Epinions}, was  $50$, $72$, $86$, $106$ minutes for DNE-SBP, SigGAN, SIDE and SNE respectively. We show our observations on \textit{Epinions} as it is the dataset with the most number of edges (around $0.8$ million edges). Therefore, although SigGAN has better performance than DNE-SBP, it requires more time to converge. We intuitively believe by exploring different discriminator models and by including structural role based information of the nodes in modified softmax, we can lower the execution time of SigGAN. We consider these as the future directions of the paper.

\section{Conclusions}~\label{s:con}
We propose a Generative Adversarial Network (GAN) based approach, SigGAN, for generating a representation of the nodes in signed networks. In signed networks, nodes with positively connected edges have similar representations whereas the representations of the nodes with negative edges have dissimilar representations. SigGAN is computationally efficient, handles imbalance in negative and positive edges and considers structural balance theory. Furthermore, it can be applied for predicting the sign of a link with high effectiveness.
Validation on $5$ datasets, like \textit{Slashdot}, \textit{Epinions}, \textit{Reddit}, \textit{Bitcoin} and \textit{Wiki-RFA} indicate the SigGAN can predict links with a high mean micro F1-score of $0.75-0.96$ which is higher than the existing state-of-the-art research works. As a future goal, we plan to extend SigGAN by exploring different discriminator functions, integration of role based information for node embedding and its applicability for directed signed networks. We, further, intend to make SigGAN applicable for temporal signed networks as most of the signed relationships formed in real life often undergo change in polarity with respect to time.

%\bibliography{biblio}
\bibliographystyle{acm}
\bibliography{sample1.bib}
\end{document}